\documentclass{article}
\usepackage{fullpage}
\usepackage[utf8]{inputenc}
\usepackage{amsmath,amsfonts,amsthm,amssymb}
\usepackage{graphicx}
\usepackage{subfigure}
\usepackage{graphics}
\usepackage{url}
\usepackage{hyperref}
\usepackage{color}

\newtheorem{theorem}{Theorem}[section]
\newtheorem{lemma}[theorem]{Lemma}

\newtheorem{definition}[theorem]{Definition}

\newtheorem{conjecture}{Conjecture}

\newcommand{\XXX}{\textcolor{red}{XXX}}

\newcommand{\demand}{\mathrm{dem}}

\newcommand{\ball}{\mathrm{ball}}
\newcommand{\len}{\mathrm{len}}
\newcommand{\eps}{\varepsilon}
\newcommand{\Hinv}{\overleftarrow{H}}
\newcommand{\Ginv}{\overleftarrow{G}}
\newcommand{\Qinv}{\overleftarrow{Q}}

\newcommand{\PathQ}{{\sc Path-Quasipartition}}
\newcommand{\dd}{\mathrm{d}}

\newcommand{\etal}{\textit{et al}.}

\newcommand{\spcut}{\text{Sparsest Cut }}

\newcommand{\mcut}{\text{Multicut }}
\newcommand{\mflow}{\mathsf{maxflow }}
\newcommand{\gap}{\mathsf{gap }}

\newcommand{\nptime}{\textsf{NP }}
\newcommand{\dtime}{\textsf{DTIME }}
\newcommand{\zpptime}{\textsf{ZPP }}

\title{Embeddings of Planar Quasimetrics into Directed $\ell_1$\\ and Polylogarithmic Approximation for Directed Sparsest-Cut}

\author{
Ken-ichi Kawarabayashi\thanks{National Institute of Informatics, 2-1-2, Hitotsubashi, Chiyoda-ku, Tokyo, Japan, \texttt{k\_keniti@nii.ac.jp}.  Supported by JSPS Kakenhi JP18H05291 and by JSPS Kakenhi JP20A402.}
\and
Anastasios Sidiropoulos\thanks{Dept.~of Computer Science, University of Illinois at Chicago, \texttt{sidiropo@uic.edu}.}
}

\begin{document}

\maketitle

\begin{abstract}
The multi-commodity flow-cut gap is a fundamental parameter that affects the performance of several divide \& conquer algorithms,
and has been extensively studied for various classes of undirected graphs.
It has been shown by Linial, London and Rabinovich \cite{linial1994geometry} and by Aumann and Rabani \cite{aumann1998log} that for general $n$-vertex graphs it is bounded by $O(\log n)$ and the Gupta-Newman-Rabinovich-Sinclair conjecture \cite{gupta2004cuts} asserts that it is $O(1)$ for any family of graphs that excludes some fixed minor.

We show that the multicommodity flow-cut gap on \emph{directed} planar graphs is  $O(\log^3 n)$.
This is the first \emph{sub-polynomial} bound for any family of directed graphs of super-constant treewidth.
We remark that for general directed graphs, it has been shown by Chuzhoy and Khanna \cite{chuzhoy2009polynomial} that the gap is $\widetilde{\Omega}(n^{1/7})$, even for directed acyclic graphs.

As a direct consequence of our result, we also obtain the first polynomial-time polylogarithmic-approximation algorithms for the Directed Non-Bipartite Sparsest-Cut, and the Directed Multicut problems for directed planar graphs, which extends the long-standing result for undirectd planar graphs by Rao \cite{Rao99} (with a slightly weaker bound). 

At the heart of our result we investigate low-distortion quasimetric embeddings into \emph{directed} $\ell_1$. 
More precisely, we construct $O(\log^2 n)$-Lipschitz quasipartitions for the shortest-path quasimetric spaces of planar digraphs,
which generalize the notion of Lipschitz partitions from the theory of metric embeddings.
This construction combines ideas from the theory of bi-Lipschitz embeddings, with tools form data structures on directed planar graphs.
\end{abstract}

\section{Introduction}

The multi-commodity flow-cut gap has been proven instrumental in the design of routing and divide \& conquer algorithms in graphs.
Bounds on this parameter generalize the max-flow/min-cut theorem, and lead to deep connections between algorithm design, graph theory, and geometry \cite{linial1994geometry,aumann1998log,arora2008euclidean}.

While the flow-cut gap for several classes of undirected graphs has been studied extensively, the case of directed graphs is poorly understood. In this work we make progress towards overcoming this limitation by showing that the flow-cut gap on $n$-vertex planar digraphs is $O(\log^3 n)$.
This is the first \emph{sub-polynomial} bound on any family of digraphs of super-constant treewidth (not directed treewidth but the abstract graphs ignoring directions are of constant treewidth).
We contrast our result with the strong lower bound due to Chuzhoy and Khanna \cite{chuzhoy2009polynomial}, who showed that for general directed graphs the gap is $\widetilde{\Omega}(n^{1/7})$, even for directed acyclic graphs. 
Thus it is highly natural to deal with planar directed graphs. 
Note that there is a long standing conjecture that for undirected planar graphs the flow-cut gap is  $O(1)$ (we will mention more details later). There is no progress for a long time. So it seems hard to improve our flow-cut gap to $O(1)$ (this improvement would imply the above mentioned conjecture). 

In order to prove our main results, we investigate low-distortion metric embeddings into directed $\ell_1$. 
More precisely, our result is obtained by proving an equivalent result in the theory of bi-Lipschitz quasimetric embeddings, which are mappings that generalize the standard theory of bi-Lipschitz metric embeddings to the asymmetric case.
The distortion bound of our quasimetric embedding implies the same bound for the standard LP relaxations of various cut problems on directed graphs.
Therefore, as a direct consequence, we also obtain the first polynomial-time polylogarithmic-approximation  algorithms for several cut problems on directed graphs, including Directed Non-Bipartite Sparsest-Cut, and Directed Multicut for directed planar graphs.

\subsection{Multi-commodity flow-cut gaps for undirected graphs}
A \emph{multi-commodity flow} instance in an undirected graph $G$ is defined by two non-negative functions: $c \colon E(G) \to \mathbb{R}$ and $d \colon V(G) \times V(G) \to \mathbb{R}$. We refer to $c$ and $d$ as the \emph{capacity} and \emph{demand} functions respectively. The \emph{maximum concurrent flow} is the maximal value $\varepsilon$ such that for every $u,v \in V(G)$, $\varepsilon\cdot d(u,v)$ can be simultaneously routed between $u$ and $v$, without violating the edge capacities. We refer to this value as $\mflow(G,c,d)$.

For every $S \subseteq V(G)$, the \emph{sparsity} of $S$ is defined as follows:
\[
\frac{\sum_{(u,v)\in E(G)} c(u,v) |\mathbf{1}_S(u) - \mathbf{1}_S(v)|}{\sum_{u,v\in V(G)} d(u,v) |\mathbf{1}_S(u) - \mathbf{1}_S(v)|},
\]
where $\mathbf{1}_S \colon V(G) \to \{0,1\}$ is the indicator for membership in $S$. The sparsity of a cut is a natural upper bound for $\mflow(G,c,d)$. The \emph{multi-commodity max-flow min-cut gap} for $G$, denoted by $\gap(G)$, is the maximum gap between the value of the flow and the upper bound given by the sparsity formula, over all multi-commodity flow instances on $G$.
The flow-cut gap on undirected graphs has been studied extensively, and several upper and lower bounds have been obtained for various graph classes. The gap is referred to as the \emph{uniform} multi-commodity flow-cut gap for the special case where there is a unit demand between every pair of vertices. 
Leighton and Rao \cite{leighton1999multicommodity} showed that the uniform flow-cut gap is $\Theta(\log n)$ in undirected graphs.
Subsequently Linial, London and Rabinovich \cite{linial1994geometry} and Aumann and Rabani \cite{aumann1998log} showed that the non-uniform multi-commodity flow-cut gap for the \spcut problem with $k$ demand pairs is upper bounded by $O(\log k)$.

The flow-cut gap immediately implies a polynomial-time approximation algorithm for Sparsest-Cut, with approximation ratio equal to the gap.
For general graphs, improved approximation algorithms have been obtained via semidefinite programming relaxations. This approach, pioneered by Arora, Rao and Vazirani \cite{ARV09}, leads to $O(\sqrt{\log k})$-approximation for the uniform case, and has been extended to $O(\sqrt{\log k} \log\log k)$-approximation for the general case by Arora, Lee and Naor \cite{arora2008euclidean}.
The latter approach relies upon embeddings of metric spaces of negative type into $\ell_1$.

Besides these there are various studies of the flow-cut gap for specific graph families. 
A central conjecture posed by Gupta, Newman, Rabinovich, and Sinclair in \cite{gupta2004cuts} asserts the following.

\begin{conjecture}[GNRS Conjecture \cite{gupta2004cuts}] \label{conj:gnrs}
The multi-commodity flow-cut gap on any family ${\cal F}$ of graphs is $O(1)$ if and only if ${\cal F}$ forbids some fixed minor.
\end{conjecture}

Arguably one of the most interesting cases of Conjecture 1, which is still open, is the case of planar graphs, which is often referred to as the planar embedding conjecture:

\begin{conjecture}[Planar Embedding Conjecture] \label{conj:planar}
The multi-commodity flow-cut gap on planar graphs is $O(1)$.
\end{conjecture}

Conjecture 2 has been verified for the case of series-parallel graphs \cite{gupta2004cuts}, $O(1)$-outerplanar graphs \cite{chekuri2006embedding}, $O(1)$-pathwidth graphs \cite{lee2009geometry}, and for some special classes of planar metrics \cite{sidiropoulos2013non}. However, 
Conjecture 2 is still wide open, and the current best gap is $O(\sqrt{\log n})$ by Rao \cite{Rao99} from 1999. Since then, no improvement has been made. 
The best-known lower bound is $2$ \cite{lee2010coarse}.

For graphs excluding any fixed minor the flow-cut gap is known to be $O(1)$ for uniform demands \cite{klein1993excluded}.

%Prior work: GNRS, related results for bounded pathwidth and series parallel graphs.

\subsection{Multi-commodity flow-cut gaps for directed graphs}

For the case of directed graphs, the flow-cut gap is defined in terms of the Directed Non-Bipartite \spcut problem which is an asymmetric variant of the \spcut problem, and is defined as follows. Let $G$ be a directed graph and let $c : E(G) \to \mathbb{R}_{\geq 0}$ be a capacity function. Let $T = \{(s_1,t_1),(s_2,t_2),\ldots,(s_k,t_k)\}$ be a set of terminal pairs, where each terminal pair $(s_i,t_i)$ has a non-negative demand $\demand(i)$. A cut in $G$ is a subset of directed edges of $E(G)$. For a cut $S \subseteq E(G)$ in $G$, let $I_S$ be the set of all indices $i \in \{1,2,\ldots,k\}$ such that all paths from $s_i$ to $t_i$ have at least one edge in $S$. Let $D(S) = \sum_{i\in I_S} \demand(i)$ be the demand separated by $S$. Let $W(S) = \frac{C(S)}{D(S)}$ be the \emph{sparsity} of $S$. The goal is to find a cut with minimum sparsity. The LP relaxation of this problem corresponds to the dual of the LP formulation of the directed maximum concurrent flow problem, and the integrality gap of this LP relaxation is the directed multi-commodity flow-cut gap. Hajiaghayi and R\"{a}cke \cite{aprxdirectedsparsestcut} showed an upper bound of $O(\sqrt{n})$ for the flow-cut gap. This upper bound on the gap has been further improved by Gupta \cite{agupta03directed}, and the current best approximation ratio is given by Agarwal, Alon and Charikar to $\tilde{O}(n^{11/23})$ in \cite{agarwal2007improved}. 

Besides these there are only a few studies of the flow-cut gap for specific (directed) graph families. 
For directed graphs whose abstract graphs are of treewidth $t$, it has been shown that the gap is at most $t \log^{O(1)} n $ by M{\'e}moli, Sidiropoulos and Sridhar \cite{memoli2016quasimetric}.
Salmasi, Sidiropoulos and Sridhar \cite{salmasi2019constant} have shown that the gap for the \emph{uniform} case is $O(1)$ on directed graphs whose abstract graphs are series-parallel, and on digraphs whose abstract graphs are of bounded pathwidth, which implies $O(\log n)$ bounds for the gap on the non-uniform case.

On the lower bound side Saks \etal~\cite{saks2004lower} showed that for general directed graphs the flow-cut gap is at least $k-\eps$, for any constant $\eps>0$, and for any $k=O(\log n/\log\log n)$. The current best lower bound is given by Chuzhoy and Khanna who showed a $\tilde{\Omega}(n^{\frac{1}{7}})$ lower bound for the flow-cut gap in \cite{chuzhoy2009polynomial}.

%We note that the dual of concurrent flow problem in directed graphs corresponds to a relaxation for the Non-Bipartite \spcut problem, and thus the integrality gap of this LP-relaxation is in fact the flow-cut gap.

\paragraph{Our main result on multi-commodity flow-cut gaps.}

In this paper we obtain the first poly-logarithmic upper bound for any family of digraphs whose abstract graph is of super-constant treewidth.
Our main result on multi-commodity flow-cut gaps is as follows.

\begin{theorem}\label{thm:gap}
The uniform multi-commodity flow-cut gap on planar digraphs is $O(\log^2 n)$.  
The non-uniform multi-commodity flow-cut gap on planar digraphs is $O(\log^3 n)$.  
\end{theorem}

Given the strong lower bound by Chuzhoy and Khanna \cite{chuzhoy2009polynomial} for general directed graphs (even for directed acyclic graphs), it is 
highly natural to deal with planar digraphs. 

This result can be considered as a first step to generalize Conjecture 2 for planar digraphs. 
Let us observe that if we could show $O(1)$ gap in Theorem \ref{thm:gap}, this would imply Conjecture 2. Moreover, 
this result extends the above mentioned long-standing result for undirectd planar graphs by Rao \cite{Rao99} to planar digraphs (with a slightly weaker bound).

Theorem \ref{thm:gap} is a direct consequence of our main technical result, which uses the theory of Lipschitz quasipartitions, and is stated formally in  Theorem \ref{thm:main_lip} (see \cite{memoli2016quasimetric,salmasi2019constant} for the reduction).

\subsection{Cut problems of directed graphs}
Better bounds on the flow-cut gap typically also imply better approximation ratios for solving cut problems. For the Directed Non-Bipartite \spcut problem the flow-cut gap upper bounds of \cite{aprxdirectedsparsestcut} and \cite{agarwal2007improved} are also accompanied by $O(\sqrt{n})$ and $O(\tilde{n}^{11/23})$ polynomial time approximation algorithms respectively. Similarly for digraphs whose abstract graph are of treewidth $t$, a $t \log^{O(1)} n $ polynomial time approximation algorithm is also provided in \cite{memoli2016quasimetric}.

Another closely related cut problem is the Directed \mcut problem which is defined as follows. Let $G$ be a directed graph and let $c : E(G) \to \mathbb{R}_{\geq 0}$ be a capacity function. Let $T = \{(s_1,t_1),(s_2,t_2),\ldots,(s_k,t_k)\}$ be a set of terminal pairs. A \emph{cut} in $G$ is a subset of $E(G)$. The \emph{capacity} of a cut $S$ is $c(S) = \sum_{e \in S}c(e)$. The goal is to find a cut separating all terminal pairs, minimizing the capacity of the cut. This problem is NP-hard. An $O(\sqrt{n \log{n}})$ approximation algorithm for Directed \mcut was presented by Cheriyan, Karloff and Rabani \cite{cheriyan2005approximating}. Subsequently an $\tilde{O}(n^{2/3}/OPT^{1/3})$-approximation was given due to Kortsarts, Kortsarz and Nutov \cite{aprxdirectedmulticut3}. Finally \cite{agarwal2007improved} also gives an improved $\tilde{O}(n^{11/23})$-approximation algorithm for this problem. Again for digraphs whose abstract graphs are of  treewidth $t$ a $t \log^{O(1)} n $ approximation algorithm was also shown in \cite{memoli2016quasimetric}.

On the hardness side \cite{chuzhoy2006hardness} demonstrated an $\Omega(\frac{\log{n}}{\log{\log{n}}})$-hardness for the Directed  Non-Bipartite \spcut\\ problem and the Directed \mcut problem under the assumption that \nptime $\not\subseteq$ \dtime $(n^{\log{n}^{O(1)}})$. This was further improved by them in a subsequent work \cite{chuzhoy2009polynomial} to obtain an $2^{\Omega(\log^{1 - \varepsilon}{n})}$-hardness result for both problems for any constant $\varepsilon > 0$ assuming that \nptime $\subseteq$ \zpptime.

Our main results for these problems are the following theorems.

\begin{theorem}\label{thm:spcutresults}
There exists a polynomial-time $O(\log^2 n)$-approximation algorithm for the Uniform Directed \spcut problem on planar digraphs. 
Moreover, there exists a polynomial-time $O(\log^3 n)$-approximation algorithm for the General Directed \spcut problem on planar digraphs. 
\end{theorem}

\begin{theorem}\label{thm:multicutresults}
There exists a polynomial time $O(\log^2 n)$-approximation algorithm for the Directed \mcut  \\ problem on planar digraphs. 
\end{theorem}

Both of the above results are direct consequences of our results on the theory of bi-Lipschitz quasimetric embeddings, which we describe in the next subsection.
We refer the reader to \cite{memoli2016quasimetric,salmasi2019constant} for a detailed description of the reduction from directed cut problems to bi-Lipschitz quasimetric embeddings.

\subsection{Quasimetric spaces and embeddings}

\paragraph{Random quasipartitions.}
A \emph{quasimetric space} is a pair $(X,d)$ where $X$ is a set of points and $d:X\times X\to \mathbb{R}_{+}\cup\{+\infty\}$, that satisfies the following two conditions:

\begin{description}
\item{(1)}
For all $x,y\in X$, $d(x,y)=0$ iff $x=y$.
\item{(2)}
For all $x,y,z\in X$, $d(x,y)\leq d(x,z)+d(z,y)$.
\end{description}

The notion of random quasipartitions was introduced in \cite{memoli2016quasimetric}.
A \emph{quasipartition} of a quasimetric space $(X,d)$ is a transitive reflexive relation on $X$.
We remark that this definition is a generalization of the notion of a partition of a set of points, which can be viewed as a symmetric transitive binary relation, where two points are related if and only if they belong to the same cluster.

Let $G$ be a digraph and let $M=(V(G),d_G)$ be its shortest-path quasimetric space. 
Let $F\subseteq E(G)$.
Let
\[
R = \{(u,v)\in V(G)\times V(G) : v \text{ is reachable from } u \text{ in } G\setminus F\}.
\]
It is immediate to check that $R$ is indeed a quasipartition.
We say that $R$ is \emph{induced by the cutset} $F$.
%Similarly, we can obtain a cutset from a quasipartition.
%Specifically, given some quasiparitition $R'$, we define 
%\[
%F' = \{(u,v)\in E(G) : (u,v)\notin R'\}.
%\]
%We say that $F'$ is the \emph{cutset induced by} $R'$.
%It is immediate to check that a quasipartition $R$ is induced by cutset $F$ if and only if $F$ is induced by $R$.
%Thus, we may occasionally treat quasipartitions and cutsets interchangeably.

Let $M=(X,d)$ be a quasimetric space. For any fixed $r \geq 0$, we say that a quasipartition $Q$ of $M$ is \emph{$r$-bounded} if for every $x,y \in X$ with $(x,y) \in Q$, we have $d(x,y) \leq r$.
For any $\beta>0$,
we say that a distribution over $r$-bounded quasipartitions of $M$, ${\cal D}$, 
is \emph{$\beta$-Lipschitz} if for any $x,y\in X$, we have that
\[
\Pr_{P\sim {\cal D}}[(x,y)\notin P] \leq \beta \frac{d(x,y)}{r}.
\]

Given a distribution $\mathcal{D}$ over quasipartitions we sometimes use the term random quasipartition (with distribution $\mathcal{D}$) to refer to any quasipartition $P$ sampled from $\mathcal{D}$.
We consider the quasimetric space obtained from the shortest path distance of a directed graph. M\'emoli, Sidiropoulos and Sridhar in \cite{memoli2016quasimetric} find an $O(1)$-Lipschitz distribution over $r$-bounded quasipartitions of tree quasimetric spaces. They also prove the existence of a $O(t\log n)$-Lipschitz distribution over $r$-bounded quasipartitions for any quasimetric that is obtained from a directed graph of treewidth $t$.

Our main result on Lipschitz quasipartitions is the following theorem.
This is the main technical contribution of this paper.
The fact that the quasipartition is efficiently samplable is needed in the algorithmic applications.

\begin{theorem}\label{thm:main_lip}
Let $n\in \mathbb{N}$,
and let $G$ be an $n$-vertex planar digraph with non-negative edge lengths.
Then for any $\Delta > 0$,
there exists a $\Delta$-bounded, $ O(\log^2 n)$-Lipschitz random quasipartition of the quasimetric space $(V(G), d_G)$.
Moreover the quasipartition is samplable in polynomial time.
\end{theorem}

%In Section \ref{sec:treewidth}, we consider graphs of treewidth 2 (series parallel graphs), and we show that for every $r >0$, we can get an $O(1)$-Lipschitz distribution over $r$-bounded quasipartitions of $M$, where $M$ is the shortest path quasimetric space corresponding to a graph of treewidth 2. In Section \ref{sec:pathwidth}, we get a similar result for the graphs of bounded pathwidht; that is for every $r > 0$, we get an $O(1)$-Lipschitz distribution over $r$-bounded quasipartitions of $M$, where $M$ is the shortest path quasimetric space corresponding to a graph of bounded pathwidth.

\paragraph{Quasimetric embeddings.}
Before stating our embedding results, we first need to introduce some notations and definitions. Let $M=(X,d)$ and $M'=(X',d')$ be quasimetric spaces.
A mapping $f:X\to X'$ is called an \emph{embedding of distortion $c\geq 1$} if there exists some $\alpha>0$, such that for all $x,y\in X$, we have 
$d(x,y)\leq \alpha \cdot d'(f(x),f(y)) \leq c\cdot d(x,y)$. 
We say that $f$ is \emph{isometric} when $c=1$.
We remark that this definition directly generalizes the standard notion of distortion of maps between metric spaces; however, here the ordering of any pair of points is important when dealing with distances.

%Let ${\cal D}$ be a distribution over pairs $(M',f)$, where $f:X\to X'$.
%We say that ${\cal D}$ is a random embedding of distortion $c\geq 1$ if for all $x,y\in X$, the following conditions are satisfied:
%\begin{description}
%\item{(1)}
% $\Pr_{(M',f)\sim {\cal D}}[d'(f(x),f(y)) \geq d(x,y)] = 1$.
%\item{(2)}
% $\mathbf{E}_{(M',f)\sim {\cal D}}[d'(f(x),f(y))]\leq c\cdot d(x,y)$.
%\end{description}

\paragraph{Directed $\ell_1$ (Charikar \etal~ \cite{charikar2006directed}).}
 The directed $\ell_1$ distance between two points $x$ and $y$ is given by $d_{\ell_1}(x,y) = \sum\limits_i |x_i -y_i| + \sum\limits_i |x_i| - \sum\limits_i |y_i| $.

We can now state the main result on  embeddability into directed $\ell_1$,
which is an immediate consequence of Theorem \ref{thm:main_lip}.  
For a detailed description of the reduction from the directed $\ell_1$ embeddability to Lipschitz quasipartitions, we refer the reader to \cite{memoli2016quasimetric,salmasi2019constant}.

\begin{theorem}\label{thm:main_embed}
Let $n\in \mathbb{N}$,
and let $G$ be an $n$-vertex planar digraph with non-negative edge lengths.
Then the quasimetric space $(V(G), d_G)$ admits an embedding into directed $\ell_1$ with distortion  $O(\log^3 n)$.
Moreover the embedding is computable in polynomial time.
\end{theorem}

%The main results of this paper are stated in Section \ref{sec:application} where we discuss some constant approximation algorithms for different cut problems. By using the results of Sections \ref{sec:treewidth} and \ref{sec:pathwidth}, we are able to find a polynomial time $O(1)$-approximation algorithm for the Directed \mcut problem on graphs of treewidth 2 as well as bounded pathwidth graphs. Furthermore, we show that for these family of graphs, there exists a polynomial time $O(1)$-approximation algorithm for the Directed \spcut problem with uniform demands. Note that these are the first constant approximation results for these problems for a non-trivial family of graphs.

%In Section \ref{sec:cyclel1} we discuss a random embedding for directed cycles into directed $\ell_1$. Let $G$ be a directed cycle and let $M=(V(G), d_G)$ be the corresponding shortest-path quasimetric space. We show that $M$ admits a constant-distortion embedding into some convex combination $D$ of 0-1 quasimetric spaces, and moreover, we can sample a random 0-1 quasimetric space from $D$ in polynomial time. Finally, in Section \ref{sec:lowerbound} we complete our paper with a lower bound result and we show that there exists a directed cycle $G = (V,E)$ such that any non-contracting random embedding of $G$ into directed trees has distortion $\Omega(n)$.

\subsection{High-level Overview of our Proof}

As we have explained above, all of the results in this paper are direct consequences of a single result concerning random quasipartitions.
Specifically, we show that for any planar digraph $G$, 
and for any ``scale'' $\Delta>0$, 
there exists a $\Delta$-bounded $O(\log^2 n)$-Lipschitz quasipartition of the quasimetric space $(V(G), d_G)$.

We remark that it is known that the classical theory of Lipschitz partitions of metric spaces is known to not be sufficient for obtaining such a result (see \cite{memoli2016quasimetric,salmasi2019constant}).
Thus, in order to handle the quasimetric case we  combine ideas from the theory of Lipschitz partitions with methods developed within the literature of distance oracles for directed planar graphs, developed by Thorup \cite{thorup2004compact}. 

Because we deal with directed graphs, some difficulties in using the distance oracle results appear in our proof. See Figure \ref{fig:path_step1} for example. Moreover, 
we also have to consider some random partitioning scheme of directed graphs, which is similar to the one used by Klein, Plotkin and Rao \cite{klein1993excluded} for undirected graphs.

We now explain what are the main ingredients used in the proof, and how they are combined in our main algorithm.
The input to the algorithm consists of a directed planar graph $G$ and $\Delta>0$.
The output is a random quasiparition of the quasimetric space $(V(G), d_G)$.
The main steps of the algorithm are as follows.
In order to simplify the exposition, we relax slightly the notation.

\begin{description}
    \item{\textbf{Step 1: Partioning $G$ into layers.}}
    It has been shown by Thorup \cite{thorup2004compact} that any planar digraph can be decomposed into a sequence of vertex-disjoint ``layers'' $L_1,L_2,\ldots$, such that each layer $L_i$ has a directed rooted spanning tree $T_i$, such that at least one of the two following properties hold:
   \begin{description}
   \item{(i)}
   Every vertex in $L_i$ can be reached from the root via a directed path of length at most $\Delta$ in $T_i$.
   \item{(ii)}
   Every vertex in $L_i$ can reach the root via a directed path of length at most $\Delta$ in $T_i$.
\end{description}
We construct a randomized version of this decomposition scheme, which is similar to the random partitioning scheme used by Klein, Plotkin and Rao \cite{klein1993excluded}.
One important property is that during this random decomposition step, each edge $(u,v)\in E(G)$ is cut with probability at most $O(d_G(u,v)/\Delta)$.
A key property is that, as in Thorup \cite{thorup2004compact},  every path of length at most $\Delta$ can be contained in at most three consecutive layers. 
This fact implies that in order to obtain a quasipartition for $G$, it is enough to compute a quasipartition for each layer, and then output their ``common refinement'' (which is formally defined by taking the union of the corresponding cutsets).

    \item{\textbf{Step 2: Quasipartioning each layer.}}
    It now remains to show how to compute a quasipartiton for each layer.
    Fix some layer $L$.
    We now invoke another result of Thorup \cite{thorup2004compact}, who showed that there exists three shortest paths $P_1$, $P_2$, $P_3$ of length at most $\Delta$ in $L$, such that $V(P_1\cup P_2\cup P_3)$ is a balanced vertex separator of $L$.
    Of course, we cannot simply delete the paths $P_1,P_2,P_3$, since this would cut some edges with probability $1$ (since the choice of the paths is deterministic).
    Instead, we show how to delete some random set of edges $C\subseteq E(L)$, such that both of the following two conditions hold for the resulting digraph:
    \begin{description}
    \item{(i)}
    Each edge $(u,v)\in E(L)$ is cut with probability at most $O(d_L(u,v)/\Delta)$.
    
    \item{(ii)}
    If there exists a path $R$ from $u$ to $v$ that intersects  $P_1\cup P_2\cup P_3$, 
    then it must be that $d_G(u,v) \leq \Delta$.
    \end{description}
    We refer to the above process as ``quasipartitioning the neighborhood of the separator''; we give more details about this process in Step 3.
    
    Equipped with this process, we easily obtain a quasipartition for $L$ via recursion: we recursively run the algorithm, effectively quasipartitioning the neighborhoods of all separators of each connected component of $L\setminus V(P_1\cup P_2\cup P_3)$.
    The final output is the union of all cutsets computed in all recursive calls.
    Since every edge can be contained in at most $O(\log n)$ recursive calls, it follows by the union bound that the probability of cutting any particular edge is at most $O(\log n)$ times the probability of cutting it in each quasipartitioning step, which we show below is at most $O(\log n \cdot d_L(u,v)/\Delta)$, which implies the result.
    
\item{\textbf{Step 3: Quasipartioning the neighborhood of a separator.}}   
Because the separators involved in Step 2 consist of the union of at most three shortest paths, it can be shown that it is enough to handle the case of a single separator path $P$.
An important difficulty here is that common intuition from the notion of a ``neighborhood'' in the case of undirected graphs, does not directly translate to the directed case that we are dealing with.
Roughly speaking, the reason is that there are two different neighborhoods of $P$: one consisting of all vertices that are at distance at most $\Delta$ \emph{from} some vertex in $P$,
and 
one consisting of all vertices that are at distance at most $\Delta$ \emph{to} some vertex in $P$ (See Figure \ref{fig:path_step1} for example).
We obtain a quasipartitioning scheme that, roughly speaking, handles the ``overlay'' of these two neighborhoods.
This process is inspired from the algorithm of Bartal \cite{bartal1996probabilistic} for computing Lipschitz partitions of undirected graphs.
Bartal's algorithm computes a random partition by growing balls of radii that are distributed according to some truncated exponential distribution.
We adapt Bartal'a algorithm by growing random ``quasiballs'', and showing that a  similar analysis can be used to analyze their behavior.
\end{description}
This completes the high-level description of our approach.

\paragraph{Comparison to Klein-Plotkin-Rao and beyond.}

It has been shown by Klein, Plotkin and Rao \cite{klein1993excluded}, that for any undirected graph $G$ that excludes some fixed minor, for any $\Delta>0$, the shortest path metric $(V(G), d_G)$ admits a $\Delta$-bounded, $O(1)$-Lipschitz partition.
Here, we obtain a slightly weaker bound (i.e.~$O(\log^2 n)$ instead of $O(1)$), for the more general case of directed graphs, but our resut holds only for planar digraphs.
The reason is that the approach in \cite{klein1993excluded} directly produces a forbidden minor, when the partitioning scheme fails to produce a bounded partition.
This is, in fact, one of the very few known results for general minor-free graph families, that does not use the celebrated theory of graph minors of Robertson and Seymour \cite{GM-series}.
In contrast, we use planarity in very specific ways in order to randomly ``decompose'' the input graph into more manageable pieces. 
More precisely, we only use the above mentioned path separator theorems by Thorup  \cite{thorup2004compact} for planar undirected graphs. Ittai and Gavoille \cite{Ittai06} showed that the above planar result can be extended to minor-free undirected graphs. 
We leave it as an open problem to extend our approach to arbitrary families of minor-free directed graphs. 

\subsection{Organization}
In Section \ref{sec:prelim} we introduce some definitions that will be used throughout the rest of the paper.
In Section \ref{sec:exp_balls} we describe a general process for deleting random subsets of edges, that is analogous to Bartal's Lipschitz partitioning scheme \cite{bartal1996probabilistic}.
In Section \ref{sec:PathQ} we show how to apply this general process to the special case of quasipartitioning the ``neighborhood'' of a shortest path. This corresponds to the above Step 3. 
In Section \ref{sec:waves} we introduce a randomized variant of Thorup's decomposition scheme \cite{thorup2004compact};
we also use (recursively) the path separator theorems from \cite{thorup2004compact} in combination with the quasipartitioning scheme from Section \ref{sec:PathQ}, to obtain a quasipartitioning scheme for each part of the decomposition. This corresponds to the above Step 1. 
Finally, we put all the ingredients together in Section \ref{sec:lip}, where we present the main algorithm for sampling a Lipschitz quasipartition in any planar digraph.

\section{Definitions and Preliminaries}
\label{sec:prelim}

We now introduce some notation that will be used throughout the paper. 
For any functions $f,g:\mathbb{N}\to \mathbb{N}$, we write $f(n) \lesssim g(n)$ whenever $f(n)=O(g(n))$.

\paragraph{Graphs and subgraphs.}
For any graph $G$ and any $U\subset V(G)$, we denote by $G[U]$ the subgraph of $G$ induced by $U$; that is $V(G[U])=U$, and $E(G[U])=E(G)\cap (U\times U)$.
For any $r>0$, we write
\[
N_G(U,r) = \{v\in V(G) : d_G(U,v)\leq r\},
\]
where
\[
d_G(U,v) := \min_{u\in U} d_G(u,v).
\]

For any digraph $G$, any $u,v\in V(G)$, and any directed path $P$ from $u$ to $v$ in $G$, we refer to $u$ and $v$ as the \emph{head} and \emph{tail} of $P$, respectively.
We denote by $\len(P)$ the length of $P$.
For any $a,b\in V(P)$, we write $a\leq_P b$ if $a$ precedes $b$ in the traversal of $P$.
We also write $b\geq_P a$ if $a\leq_P b$.

For any graph $G$, $X\subseteq V(G)$, and $\alpha\in [1,0]$, we say that $X$ is a \emph{$\alpha$-balanced vertex separator} of $G$ if every connected component of $G\setminus X$ has at most $\alpha |V(G)|$ vertices.

For any digraph $G$, let $\Ginv$ denote the digraph obtained by reversing the direction of every edge.
For any path $Q$ in $G$, let $\Qinv$ denote the path in $\Ginv$ obtained by reversing the direction of every edge in $Q$.

For any digraph $G$, and any path $P$ in $G$, and any $x,y\in V(P)$, such that either $x=y$, or $x$ appears before $y$ in $P$,
we denote by $P[x,y]$ the subpath of $P$ from $x$ to $y$.
Adapting standard set-theoretic notation, we also write
$P[x,y) = P[x,y]\setminus \{y\}$,
$P(x,y) = P[x,y]\setminus \{x\}$,
and
$P(x,y) = P[x,y]\setminus \{x,y\}$,

%\paragraph{Quasimetric spaces and Lipschitz quasipartitions.}

\iffalse
\section{Triangulations}

\begin{lemma}
Let $G$ be a planar digraph.
Then there exists a triangulated planar digaph $H$, 
with all edges having lengths $0$ or $1$,
and an embedding $f:V(G)\to V(H)$ with distortion at most $\XXX$.
\end{lemma}

\begin{proof}
\XXX
\end{proof}

\fi

\section{Exponential Random Quasiballs}
\label{sec:exp_balls}

In this section we analyze the behavior of a general process for constructing quasipartitions.
The process involves computing a sequence of cutsets that are induced by quasiballs of exponentially distributed radius.

The process is completely analogous to the algorithm used by Bartal~\cite{bartal1996probabilistic} to compute probabilistic partitions of metric spaces.

Let $H$ be a digraph, with $|V(H)|=n$, and let $\Delta > 0$.
We define a probability distribution ${\cal D}$ over cutsets of $H$, i.e.~$2^{E(H)}$, as follows.
First, we define a sequence of pairwise disjoint subsets $B_1,\ldots,B_t\subseteq V(H)$.
Let $I=[0, \Delta \ln n)$, and let ${\cal I}$ be the truncated exponential probability distribution on $I$ with probability density function $p(x) = \frac{n}{n-1} \cdot \frac{1}{\Delta} \cdot e^{-x/\Delta}$.
Initially all vertices in $H$ are \emph{unmarked}.
Let $v_1\in V(H)$ be chosen arbitrarily.
Let $R_1 \sim {\cal I}$, and define
\[
B_1 = \ball_H(v_1, R_1).
\]
We mark all vertices in $B_1$.
For any $i\geq 1$, we proceed inductively.
%Suppose that $B_{i-1}$ has been defined.
If all vertices in $H$ are marked, then we terminate the sequence at $t=i-1$.
Otherwise, we arbitrarily either terminate the sequence at $t=i-1$, or we proceed to define $B_i$; that is, this choice can depend on all previous random choices.
We set $v_i$ be to an arbitrary unmarked vertex.
Let $R_i\sim {\cal I}$, and define 
\[
B_i = \ball_H(v_i, R_i) \setminus \left(B_1\cup \ldots \cup B_{i-1}\right).
\]
We mark all vertices in $B_i$.
This completes the inductive construction of the sequence $B_1,\ldots,B_t$.
Finally, we set
\[
F = \bigcup_{i=1}^t E(H) \cap \left( B_i \times (V(H) \setminus (B_1\cup \ldots \cup B_i)) \right).
\]
That is, for all $i\in \{1,\ldots,t\}$, 
the cutset $F$ contains all edges starting in $B_i$ and ending at outside of $B_1\cup \ldots \cup B_i$.
This completes the definition of the random cutset $F$.
We say that the resulting probability distribution ${\cal D}$ is a \emph{$\Delta$-shock}.

%, and thus also completes the definition of the distribution ${\cal D}$.

We remark that the above process is quite general and allows for many different strategies for choosing the next vertex $v_i$, and for deciding whether to terminate the sequence before all vertices have been marked.
In our main algorithm, we will apply this process by choosing the vertices $v_1,\ldots,v_t$ from an inverse traversal of some shortest path (see Section \ref{sec:PathQ}).

We next analyze $\Delta$-shocks.

\begin{lemma}\label{lem:exp_balls}
Let $H$ be a digraph,
let $\Delta > 0$, 
let ${\cal D}$ be a $\Delta$-shock,
and let
$F\sim {\cal D}$.
Then, for any $(u,v)\in E(H)$, we have
\[
\Pr[(u,v)\in F]  \leq 2 d_H(u,v)/\Delta.
\]
\end{lemma}

The proof of Lemma \ref{lem:exp_balls} is similar to the analysis of an analogous result for the case of metric cases from Bartal \cite{bartal1996probabilistic}.
For the sake of completeness, the proof of Lemma \ref{lem:exp_balls} is given in Section \ref{app:exp_balls}.

\section{Quasipartitioning the Neighborhood of a Directed Path}
\label{sec:PathQ}

%\paragraph{Procedure Path-Partition}
Let $H$ be a digraph,
let $\Delta>0$,
and let $Q$ be a shortest path in $H$ from some $v\in V(H)$ to some $v'\in V(H)$,
of length at most $\Delta$.
We describe a procedure for computing a random quasipartition, given $H$ and $Q$.
Let $I=[0, \Delta \ln n)$, and let ${\cal I}$ be the truncated exponential probability distribution on $I$ with probability density function $p(x) = \frac{n}{n-1} \cdot \frac{1}{\Delta} \cdot e^{-x/\Delta}$.
The procedure consists of the following steps.

\begin{description}
\item{\textbf{Procedure \PathQ}$(H, Q, \Delta)$}
\item{\textbf{\text{Step 1: Cutting balls away from $Q$.}}}
For any $i\in \mathbb{N}$, let $R_i \sim {\cal I}$.
Let $a_1 = v'$, 
\[
B_1 = \ball_H(a_1, R_1),
\]
For any $i>1$, let 
$a_i$ be the first vertex not in $B_1\cup \ldots \cup B_{i-1}$ that we visit when traversing $\Qinv$ starting from $v'$.
Let
\[
B_i = \ball_H(a_i, R_i) \setminus \left(\bigcup_{j=1}^{i-1} B_j\right).
\]
For any $i\in \mathbb{N}$, let 
\[
F_i = E(H) \cap (B_i\times (V(H)\setminus B_i)).
\]
See Figure \ref{fig:path_step1} for an example.

\item{\textbf{\text{Step 2: Cutting balls towards $Q$.}}}
For any $i\in \mathbb{N}$, let $R'_i \sim {\cal I}$.
Let $a'_1 = v$, 
\[
B'_1 = \ball_{\Hinv}(a'_1, R'_1),
\]
For any $i>1$, let 
$a'_i$ be the first vertex not in $B'_1 \cup \ldots \cup B'_{i-1}$ that we visit when traversing $Q$ starting from $v$;
w.l.o.g. we may assume that $d_{\Hinv}(a_{i}',a_{i-1}')=R_i'$, by  subdividing one edge of $H$.
Let
\[
B'_i = \ball_{\Hinv}(a'_i, R'_i) \setminus \left(\bigcup_{j=1}^{i-1} B'_j\right).
\]
For any $i\in \mathbb{N}$, let 
\[
F'_i = E(H) \cap ((V(H)\setminus B'_i) \times B'_i).
\]

\item{\textbf{Step 3: Output.}}
Let 
$F = \bigcup_{i\in \mathbb{N}} F_i$,
$F' = \bigcup_{i\in \mathbb{N}} F'_i$.
We output the cutset 
$C = F\cup F'$,
and the quasipartition $Z$ induced by $C$.
\end{description}
This completes the description of the procedure.

\begin{figure}
    \centering
    \scalebox{0.75}{\includegraphics{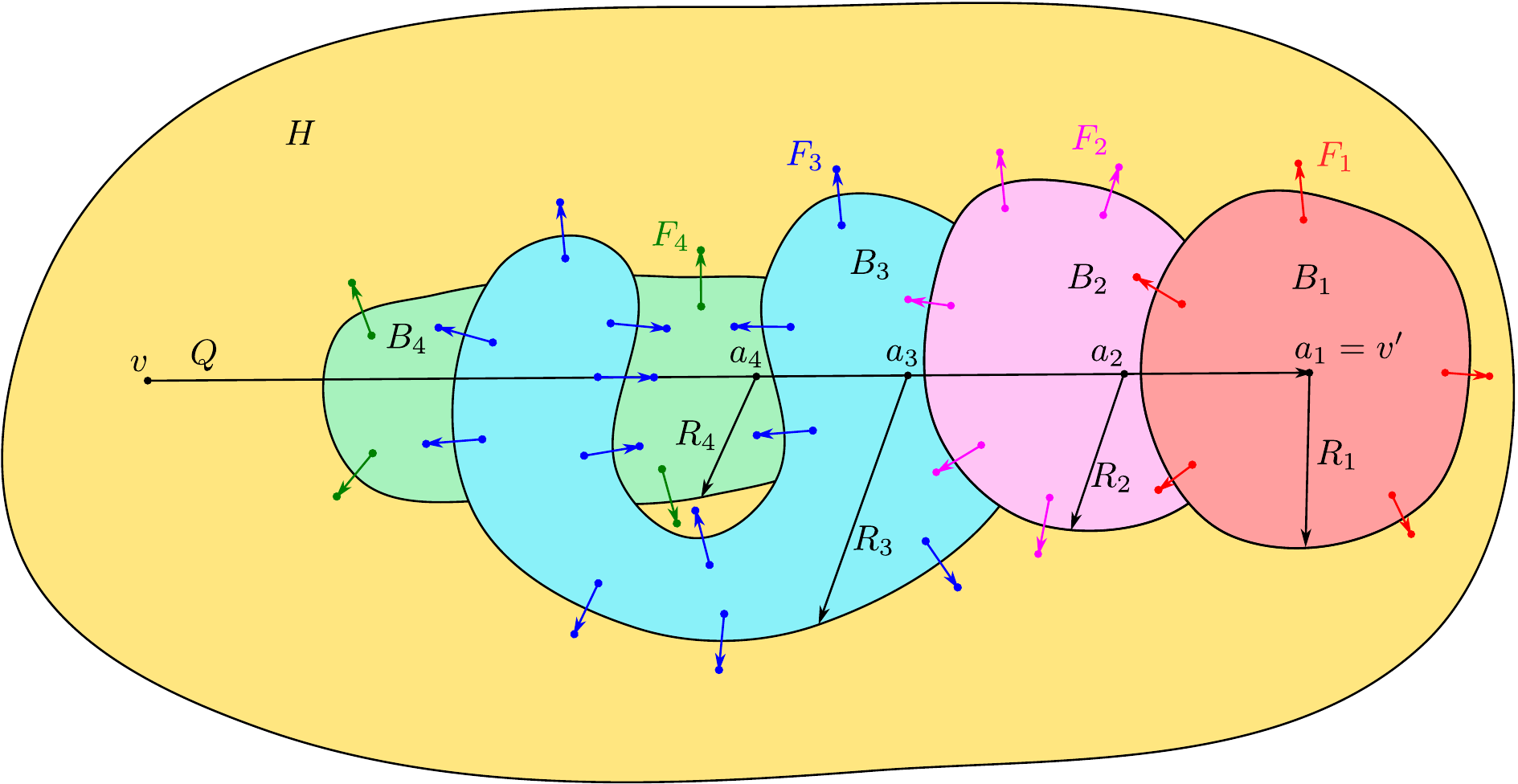}}
    \caption{An example of Step 1 in procedure \PathQ$(H, Q, \Delta)$. Each $B_i$ is depicted as a region with a distinct color. For clarity, only the regions $B_1,\ldots,B_4$ are depicted. 
    The sets $F_1, \ldots, F_4$ are also depicted, each with a different color.
    We remark that $Q$ is a shortest path from $v$ to $v'$, but does not need to be a shortest path from $v'$ to $v$. Consequently, the intersection of each $B_i$ with $Q$ does not need to consist of a single contiguous sub-path (e.g.~$B_3$).}
    \label{fig:path_step1}
\end{figure}

\paragraph{Portal assignment.}
The procedure \PathQ~computes sets of vertices $A=\{a_i\}_{i}$, $A'=\{a'_i\}_i$. 
We refer to the vertices in $A\cup A'$ as \emph{portals}.
We define mappings
\[
\pi_1 : V(Q) \to A,
\]
and
\[
\pi_2 : V(Q) \to A',
\]
as follows.
For any $a_i\in A$, and for any $x\in V(\Qinv[a_i,a_{i+1}))$,
we set
$\pi_1(x) = a_i$.
For any $a'_i\in A'$, and for any $x\in V(Q[a_i,a_{i+1}))$,
we set
$\pi_2(x) = a_i$.
Figure \ref{fig:portals} depicts an example.

\begin{figure}
    \centering
    \scalebox{0.72}{\includegraphics{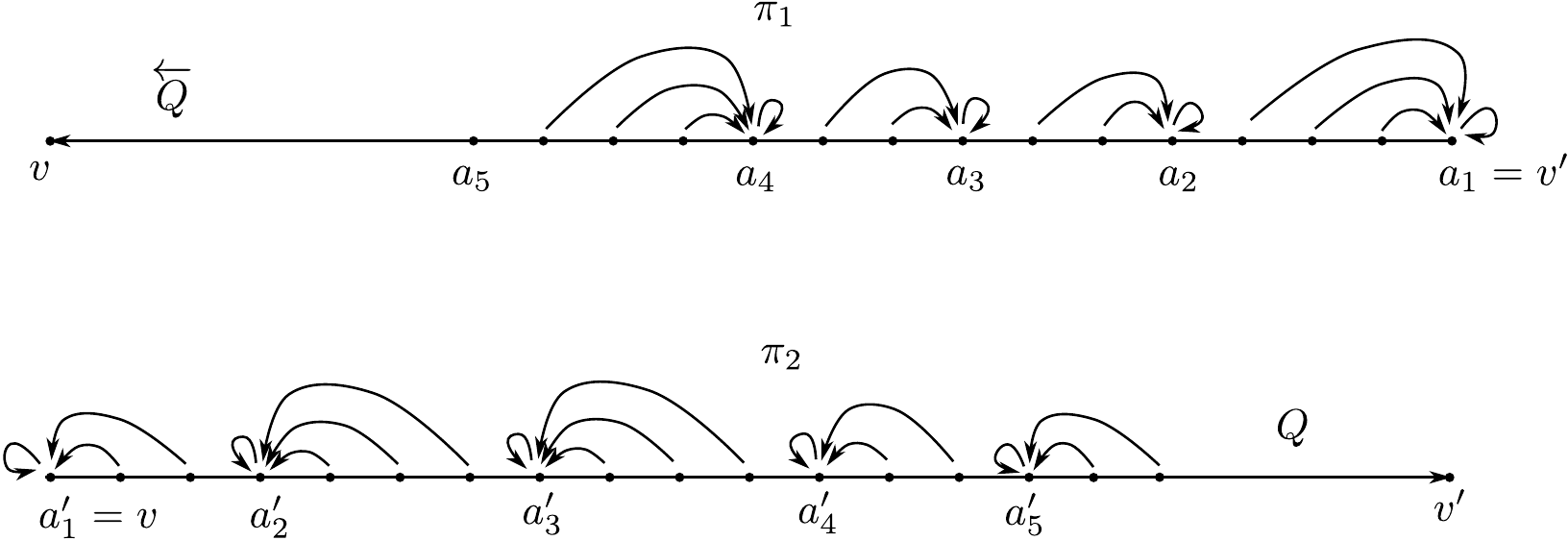}}
    \caption{An example portal assignment.}
    \label{fig:portals}
\end{figure}

\iffalse
\begin{lemma}\label{lem:portals_near}
For any $x\in V(Q)$, we have
$d_H(\pi_1(x), x) \lesssim  \Delta \log n$,
and
$d_H(x, \pi_2(x)) \lesssim \Delta \log n$.
\end{lemma}

\begin{proof}
For any $a_i,a_{i+1}\in A$, we have by construction that $a_{i+1}\leq_Q a_i$ (see Figure \ref{fig:path_step1 }).
Let $x\in V(Q)$, and let $j\in \mathbb{N}$, such that $\pi_1(x)=a_j$.
If $a_j$ is the last portal computed in Step 1, then $d_H(a_j,x)\leq \len(Q) \leq \Delta$, and the assertion follows.
Thus, it remains the consider the case where $a_j$ is not the last portal.
By the portal assignment, we have $x\in Q(a_{j+1}, a_j]$.
When the portal $a_{j+1}$ is chosen in Step 1, all vertices in $Q(a_{j+1}, a_j]$ have already been clustered, i.e.~
\[
V(Q(a_{j+1}, a_j]) \subseteq \bigcup_{t=1}^j B_t.
\]
Therefore, there exists $\ell\leq j$, such that $x\in B_\ell$.
It follows that 

\begin{align}
d_{H}(\pi_1(x),x) &= d_H(a_j,x) \notag \\
 &\leq d_H(a_j,a_{\ell}) + d_H(a_{\ell},x) & \text{(triangle inequality)} \notag \\
 &\leq \len(Q) + R_{\ell} & \text{($a_j\leq_Q a_\ell$, $x\in B_{\ell}\subseteq \ball_H(a_{\ell},R_{\ell})$)} \notag \\
 &\lesssim \Delta \log n, \notag 
\end{align}
which proves the assertion.

By a completely identical argument on $\Hinv$, we obtain that $d_{\Hinv}(\pi_2(x),x) \lesssim \delta \log n$, 
which is equivalent to
$d_{H}(x, \pi_2(x)) \lesssim \delta \log n$.
This proves the second part of the assertion, and completes the proof.
\end{proof}
\fi

\paragraph{Analysis.}
We now analyze the quasipartition computed by procedure \PathQ.

\begin{lemma}\label{lem:path_q}
Let $H$ be a digraph, let $Q$ be a shortest path in $H$ of length at most $\Delta>0$.
Let $C$ be a random cutset computed by Procedure \PathQ$(H,Q,\Delta)$, and let $Z$ be the induced quasipartition.
Then the following conditions hold:
\begin{description}
\item{(1)}
For any $(x,y)\in E(H)$, 
\[
\Pr[(x,y)\notin Z] \lesssim  \frac{d_H(x,y)}{\Delta}.
\]

%\item{(2)}
%Let $x\in V(H)$, such that $x$ can reach some $y\in V(Q)$ in $H\setminus F'$. Then $d_H(x,\pi_2(y)) \leq \XXX \Delta$.

%\item{(3)}
%Let $x\in V(H)$, such that some $y\in V(Q)$ can reach $x$ in $H\setminus F$. Then $d_H(\pi_1(y),x) \leq \XXX \Delta$.

\item{(2)}
Let $x,y\in V(H)$.
Suppose that there exists some path $P$ from $x$ to $y$ in $H\setminus C$, such that $P$ intersects $Q$.
Then, $d_H(x,y) \lesssim \Delta \log n$.
\end{description}
\end{lemma}

\begin{proof}
First we prove part (1).
From Lemma \ref{lem:exp_balls} it follows that for any $i\in \{1,2\}$, during Step $i$, the probability of cutting $(x,y)$ is at most $2 d_H(x,y)/\Delta$.
By taking a union bound over the two steps, the assertion follows.

It remains to prove part (2).
Let $a$ and $b$ be the first and last vertices in $V(P)\cap V(Q)$ that we encounter along a traversal of $P$.

We argue that $d_H(x,\pi_2(a)) \lesssim \Delta \log n.$
Recall that $C=F\cup F'$, where $F$ is the cutset computed in Step 1, and $F'$ is the cutset computed in Step 2.
Let $P_1=P[x,a]$.
Since $P\subseteq H\setminus C$, it follows that $E(P_1)\cap F'=\emptyset$.
By induction on $i\in \mathbb{N}$, we have that for all $i\geq 1$, there exists no path in $H\setminus F'$ from any vertex in $V(H)\setminus (B'_1\cup \ldots \cup B'_i)$ to any vertex in $B'_1\cup \ldots \cup B'_i$.
Therefore, there exist $j,k\in \mathbb{N}$, with $j\leq k$, such that $x\in B'_j$ and $a\in B'_k$ (see Figure \ref{fig:P1_not_cut}).
%This implies that $\pi_2(x)\leq_Q \pi_2(a)$.
Thus
\begin{align}
d_H(x,\pi_2(a)) &\leq d_H(x,a_j) + d_H(a_j,\pi_2(a)) & \text{(triangle inequality)} \notag \\
 &\lesssim \Delta \log n + d_H(a_j, \pi_2(a)) & (x\in B'_j \subseteq \ball_{\Hinv}(a_j, R_j)) \notag \\
 &\lesssim \Delta \log n + \Delta & \text{($a_j \leq_Q a_k \leq_Q \pi_2(a)$ and $\len(Q)\leq \Delta$)} \notag \\
 &\lesssim \Delta \log n, \label{eq:connect_1}
\end{align}
as required.

We next argue that $d_H(\pi_1(b), y)  \lesssim \Delta \log n.$
Arguing as in the paragraph above, we have that there exist $s,t\in \mathbb{N}$, with $s\geq t$, such that $b\in B_s$ and $y\in B_t$.
%$\pi_1(b)\leq_Q \pi_1(y)$.
Thus
\begin{align}
d_H(\pi_1(b), y) &\leq d_H(\pi_1(b), a_t) + d_H(a_t, y) & \text{(triangle inequality}) \notag \\
 &\lesssim d_H(\pi_1(b), a_t) + \Delta \log n & (y\in B_t \subseteq \ball_H(a_t,R_t)) \notag \\
 &\lesssim \Delta + \Delta \log n & \text{($\pi_1(b) \leq_Q a_s \leq_Q a_t$ and $\len(Q)\leq \Delta$)} \notag \\
 &\lesssim \Delta \log n \label{eq:connect_2},
\end{align}
as required.

We next argue that $d_H(\pi_2(a), \pi_2(b)) \lesssim \Delta$.
If $\pi_2(a) \leq_Q \pi_2(b)$, then $\pi_2(a)$ can reach $\pi_2(b)$ by following $Q$, and thus $d_H(\pi_2(a), \pi_2(b))\leq \len(Q) \leq \Delta$.
Thus it remains to consider the case $\pi_2(a) \geq_Q \pi_2(b)$.
Let $i,j\in \mathbb{N}$, such that $\pi_2(a)=a'_i$ and $\pi_2(b)=a'_j$.
Since $a'_i\geq_Q a'_j$, it follows that $i\geq j$.
If $i=j$, then $\pi_2(a)=\pi_2(b)$, and thus $d_H(\pi_2(a),\pi_2(b))=0$.
Thus, it remains to consider the case $i>j$.
By induction on the execution of Step 2, we have that for all $t\in \mathbb{N}$, any vertex in $V(H)\setminus (B'_1\cup \ldots \cup B'_t)$ cannot reach any vertex in $B'_1\cup \ldots \cup B'_t$ in $H\setminus F'$.
Therefore, $a$ cannot reach $b$, which is a contradiction.
We have thus established that 
\begin{align}
d_H(\pi_2(a), \pi_2(b))\leq \Delta. \label{eq:connect_3}
\end{align}

Next, we argue that $d_H(\pi_2(b),\pi_1(b)) \leq  \Delta$.
By the choice of portals, we have that $\pi_2(b)\leq_Q b \leq_Q \pi_1(b)$.
Since $\len(Q)\leq \Delta$, this implies that 
\begin{align}
    d_H(\pi_2(b),\pi_1(b)) &\leq \len(Q) \leq \Delta. \label{eq:connect_4}
\end{align}

Combining \eqref{eq:connect_1}, \eqref{eq:connect_2}, \eqref{eq:connect_3} and \eqref{eq:connect_4} with the triangle inequality, we obtain
\begin{align}
d_H(x,y) &\leq d_H(x,\pi_2(a)) + d_H(\pi_2(a), \pi_2(b)) + d_H(\pi_2(b), \pi_1(b)) + d_H(\pi_1(b), y) \notag \\
 &\lesssim \Delta \log n + \Delta + \Delta + \Delta \log n \notag \\
 &\lesssim \Delta \log n,
\end{align}
which concludes the proof.
\end{proof}

\begin{figure}
    \centering
    \includegraphics{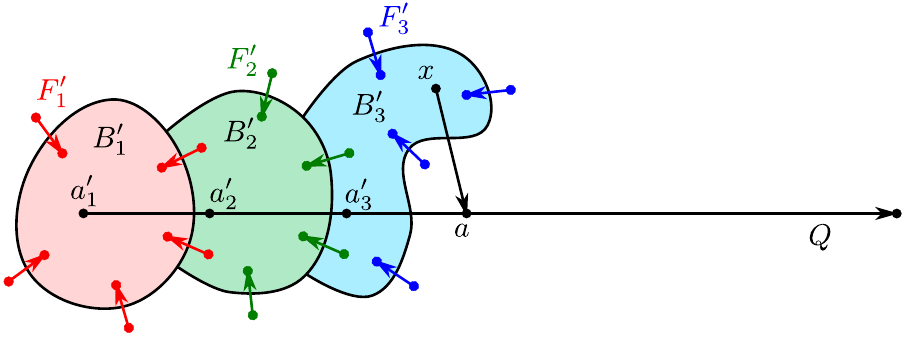}
    \caption{The configuration in the proof of Lemma \ref{lem:path_q}. The path $P_1$ is not cut by $F'$.}
    \label{fig:P1_not_cut}
\end{figure}

\section{Crashing Waves}
\label{sec:waves}

In this section, we refine some important concepts in \cite{thorup2004compact}. This is one of the keys in our analysis in the next section. 

We recall the following definition from \cite{thorup2004compact}.
A \emph{$(t,\Delta)$-layered} spanning tree $T$ in a digraph $H$ is a directed rooted spanning tree such that any root-to-leaf path in $T$ is the concatenation of at most $t$ shortest paths in $H$, each of length at most $\Delta$.
We say that $H$ is \emph{$(t,\Delta)$-layered} if it contains such a spanning tree.
We are interested in $(1,\Delta)$-layered graphs and their properties.
However, we use the slightly more general definition in order to maintain a compatible notation with   \cite{thorup2004compact}.

\iffalse
\begin{lemma}[Thorup \cite{thorup2004compact}]
Given a digraph $G$ and some $\alpha>0$, we can construct in linear time a sequence of digraphs $G_1,\ldots,G_k$, such that the following conditions are satisfied:
\begin{description}
\item{(1)}
$\sum_{i=1}^k |E(G_i)|+|V(G_i)| \lesssim |E(G)| + |V(G)|$.

\item{(2)}
For any $v\in V(G)$, there exists some index $\iota(v)$, such that for any $w\in V(G)$, $d:=d_G(v,w)\leq \alpha$, if and only if $d$ is the smallest distance from $v$ to $w$ in $G_{\iota(v)-2}, G_{\iota(v)-1}, G_{\iota(v)}$.

\item{(3)}
Each $G_i$ is a $(3,\alpha)$-layered digraph with a $(3,\alpha)$-layered spanning tree $T_i$ rooted at $r_i$.

\item{(4)}
Each $G_i$ is a minor of $G$.
\end{description}
\end{lemma}

We now derive a variant of the above result.

\begin{lemma}
Given a digraph $G$ and some $\alpha>0$, we can construct in linear time a sequence of digraphs $G_1,\ldots,G_k$, such that the following conditions are satisfied:
\begin{description}
\item{(1)}
Each $G_i$ is a minor of $G$.

\item{(2)}
Each $e\in E(G)$ is contained in at most three graphs $G_i$.

\item{(3)}
For any $v\in V(G)$, there exists some index $\iota(v)$, such that for any $w\in V(G)$, $d:=d_G(v,w)\leq \alpha$, if and only if $d$ is the smallest distance from $v$ to $w$ in $G_{\iota(v)-2}, G_{\iota(v)-1}, G_{\iota(v)}$.

\item{(4)}
Each $G_i$ is a $(3,\alpha)$-layered digraph with a $(3,\alpha)$-layered spanning tree $T_i$ rooted at $r_i$.
\end{description}
\end{lemma}
\fi

\begin{definition}[Wave]\label{defn:wave}
Let $G$ be a digraph, $v\in V(G)$, $\Delta>0$.
We define a probability distribution, ${\cal D}$, supported over $2^{E(G)}$.
Let $\tau_0,\ldots,\tau_n\in [0,\Delta)$ be chosen uniformly and independently at random.
Let 
\[
V_0 = \{u\in V(G) : d_G(v,u) \leq 2\Delta+\tau_0\},
\]
and for any $i\geq 1$, let
\[
V_i = \left\{\begin{array}{ll}
\{u\in V(G) : d_G(V_{i-1}, u) \leq 2\Delta+\tau_i\} & \text{ if $i$ is even}\\
\{u\in V(G) : d_G(u, V_{i-1}) \leq 2\Delta+\tau_i\} & \text{ if $i$ is odd}
\end{array}\right.
\]
For any $i\in \mathbb{N}$, let
\[
E_i = \left\{\begin{array}{ll}
E(G) \cap (V_{i-1} \times V_i) & \text{ if $i$ is odd}\\
E(G) \cap (V_i \times V_{i-1}) & \text{ if $i$ is even}
\end{array}\right.
\]
Finally, let
\[
E' = \bigcup_{i} E_i.
\]
We define ${\cal D}$ to be the distribution of the random variable $E'$.
We refer to ${\cal D}$ as a \emph{$(v,\Delta)$-wave}.
\end{definition}

Figure \ref{fig:wave} gives an example of the construction of a wave.

\begin{figure}
    \centering
    \scalebox{0.7}{\includegraphics{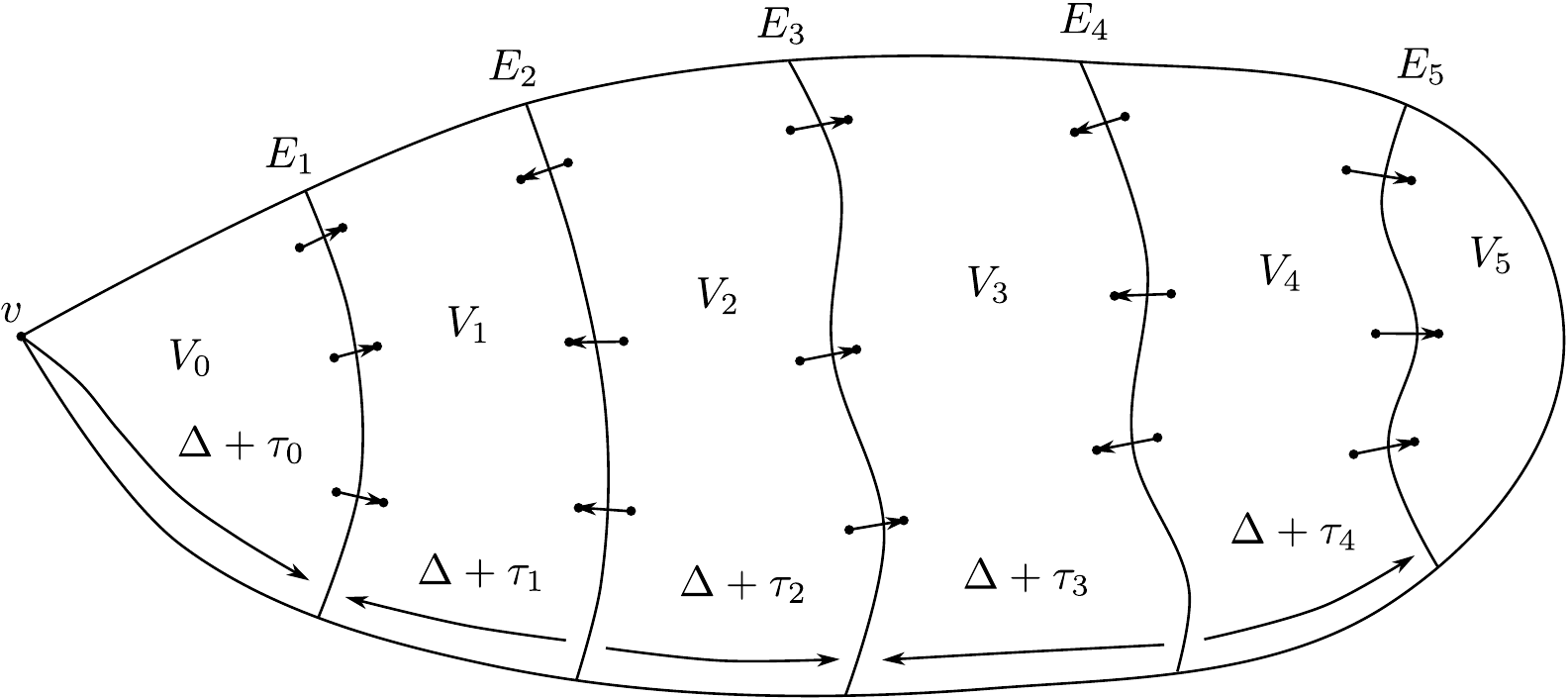}}
    \caption{A sample from a $(v,\Delta)$-wave.}
    \label{fig:wave}
\end{figure}

From the above definition, we show the following lemma. 

\begin{lemma}\label{lem:wave}
Let $G$ be a digraph, $v\in V(G)$, $\Delta>0$.
Suppose that all edges in $G$ have length at most $\Delta$.
Let ${\cal D}$ be a $(v,\Delta)$-wave.
Let $E'$ be a subset of edges sampled from ${\cal D}$, i.e.~$E'\sim{\cal D}$.
Then the following properties hold:
\begin{description}
\item{(i)}
For any $(a,b)\in E(G)$, we have that
\[
\Pr[(a,b)\in E'] \lesssim \frac{d_G(a,b)}{\Delta}.
\]
\item{(ii)}
The following holds with probability 1.
Let $P$ be any path in $G$ of length at most $\Delta$.
There exists some $i\geq 0$, such that $V(P)\subseteq V_i\cup V_{i+1}\cup V_{i+1}$.
Furthermore, $P$ can be decomposed into consecutive subpaths $P=P_1\circ P_2 \circ P_3$, such that for all $i\in \{1,2,3\}$, there exists $j\in \{i,i+1,i+2\}$, such that $P_i\subseteq V_j$

%\textcolor{red}{TODO: This needs to be rewritten so that it matches Thorup.}

\item{(iii)}
The following holds with probability 1.
For any $i\in \mathbb{N}$, contracting $\bigcup_{j<i} V_j$ into a single vertex, and deleting $\bigcup_{j>i} V_{j}$ results in a $(1,\Delta)$-layered graph.
%$G[V_i]$ is a $(1,\Delta)$-layered graph.
\end{description}
\end{lemma}

\begin{proof}
We first prove (i).
Let $(a,b)\in E(G)$.
Consider the process of sampling $E'$ from ${\cal D}$.
Let $U_0=\{v\}$,
and for any  $i\geq 1$,
let $U_i = V_0\cup \ldots \cup V_{i-1}$.
Intuitively, we can think of the process of sampling $E'$ as inductively constructing the sequence $\{V_i\}_{i\geq 0}$, and $U_i$ is the set of vertices that have already been explored just before constructing $V_i$. 
In order to cut the edge $(a,b)$ during the construction of $V_i$, it must be that 
\begin{align}
d_G(U_i, a)\leq 2\Delta+\tau_i < d_G(U_i,b). \label{eq:wave_0}
\end{align}
This implies that 
\begin{align}
d_G(U_i, a) \leq 3\Delta. \label{eq:wave_1}
\end{align}
By the definition of a wave we have that $N_G(U_i,2\Delta)\subseteq U_{i+2}$.
Combining with \eqref{eq:wave_1} we get 
\begin{align}
d_G(U_{i+2}, a) \leq \Delta. \label{eq:wave_2}
\end{align}
However, \eqref{eq:wave_2} implies that, with probability $1$, $\{a,b\}\subseteq U_{i+1}$,
and thus
$(a,b)$ cannot be cut during the construction of $V_{i+2}$.
Therefore, $(a,b)$ can potentially be cut in at most two consecutive steps during the constructing of some $V_i$ and $V_{i+1}$.
Note that $i$ is a random variable that depends on the choice of the random variables $\{\tau_i\}_{i\geq 0}$.
In each one of these two steps, by triangle inequality, we have that
$(a,b)$ is cut when \eqref{eq:wave_0} holds, which happens with probability at most $\frac{d_G(U_i,b)-d_G(U_i,a)}{\Delta} \leq \frac{d_G(a,b)}{\Delta}$.
By the union bound over these two steps, it follows that the probability that $(a,b)$ is cut at any point during the construction of the wave is at most $2\frac{d_G(a,b)}{\Delta}$, as required.
This proves part (i).

Next, we prove part (ii).
Let $P$ be any path in $G$ of length at most $\Delta$.
Let $w$ and $w'$ be the first and last vertices in $P$.
Let $i$ be the minimum integer such that $V(P)\cap V_i \neq \emptyset$.
Let $v$ be an arbitrary vertex in $V(P)\cap V_i$.
Since $P$ has length at most $\Delta$, it follows that for any $v'\in V(P)$, either $d_G(v,v')\leq \Delta$, or $d_G(v',v) \leq \Delta$.
In both cases, this implies that $v'\in V_i\cup V_{i+1}\cup V_{i+2}$.
Since $v'$ is arbitrary, we get $P\subseteq V_i\cup V_{i+1}\cup V_{i+2}$.
It remains to show that $P$ can be decomposed into consecutive subpaths $P=P_1\circ P_2 \circ P_3$, such that for all $i\in \{1,2,3\}$, there exists $j\in \{i,i+1,i+2\}$, such that $P_i\subseteq V_j$.
By the construction of a wave, it follows that for all $t$, we either cut all edges from $V_t$ to $V_{t+1}$, or we cut all the edges from $V_{t+1}$ to $V_t$.
Thus, while traversing $P$, if we leave some set $V_t$, we can never return to $V_t$ (see Figure \ref{fig:wave}).
This implies that the intersection of $P$ with each of the sets $V_i$, $V_{i+1}$, and $V_{i+2}$ is either empty or a subpath of $P$, which concludes the proof of part (ii).

Finally, part (iii) follows by contracting $U_i$ into a single vertex $u_i$, and setting the spanning tree $T$ in the definition of layered graph to be the shortest path tree rooted at $u_i$.
\end{proof}

We need the following result from \cite{thorup2004compact}. 

\begin{lemma}[Thorup \cite{thorup2004compact}]\label{lem:thorup_3_paths}
Let $\Delta>0$, and let $H$ be a planar $(1,\Delta)$-layered digraph.
Then, there exist shortest paths $P_1,P_2,P_3$ in $H$, each having length at most $\Delta$, and such that
$V(P_1)\cup V(P_2)\cup V(P_3)$ is a $2/3$-balanced vertex separator in $H$.
Moreover, there exists a polynomial-time algorithm which given $H$ and $\Delta$ outputs $P_1$, $P_2$, and $P_3$.
\end{lemma}

\begin{lemma}\label{lem:layer_lip}
Let $\Delta > 0$.
Let $H$ be a planar $(1,\Delta)$-layered digraph.
Then there exists a $\Delta$-bounded, $O(\log^2 n)$-Lipschitz quasipartition of the quasimetric space $(V(H), d_H)$.
\end{lemma}

\begin{proof}
We construct a random cutset $C$ recursively.
If $H$ contains at most $c=O(1)$ vertices, then the result follows by letting $C$ be the set of all edges in $H$ of length at least $\Delta$, with probability $1$.
Otherwise, we proceed as follows.
By Lemma \ref{lem:thorup_3_paths}, we compute shortest paths $P_1,P_2,P_3$ in $H$, each of length at most $\Delta$, such that $V(P_1\cup P_2\cup P_3)$ is a $2/3$-balanced vertex separator in $H$.
For each $i\in \{1,2,3\}$, we run procedure \PathQ$(H, P_i, \Delta)$, and obtain a random cutset $C_i\subseteq E(H)$.
We then recurse on each connected component of $H\setminus (P_1\cup P_2\cup P_3)$.
Let ${\cal P}$ be the set of all paths computed by Lemma \ref{lem:thorup_3_paths} in all recursive calls.

Let also ${\cal C}$ be the set of all random cutsets computed during all recursive calls.
We set
\[
C = \bigcup_{C'\in {\cal C}} C',
\]
and we output the quasipartition $R$ induced by $C$.
This completes the description of the algorithm.

\medskip

We next show that $R$ is $O(\log n)$-Lipschitz.
Let $(u,v)\in E(H)$.
By construction, we have that $(u,v)$ is contained in at most one subgraph at each level of the recursion.
Since the depth of the recursion is $k=O(\log n)$, it follows that $(u,v)$ is contained in at most $k$ subgraphs $H_1,\ldots,H_k$ of $H$, with 
\[
H = H_1 \supseteq H_2 \supseteq \ldots \supseteq H_k.
\]
For each $i\in \{1,\ldots,k\}$ let $P_{i,1},P_{i,2},P_{i,3}$ be the shortest paths that the algorithm computes in $H_i$, such that $V(P_1\cup P_2\cup P_3)$ is a $2/3$-balanced separator in $H_i$.
For each $i\in \{1,\ldots,k\}$ let $C_i$ be the random cutset that the algorithm computes in $H_i$.
We have
\begin{align}
\Pr[(u,v)\notin R] &\leq \Pr[(u,v)\in C] \notag \\
 &= \Pr\left[(u,v)\in C_1\cup \ldots \cup C_k\right] \notag \\
 &\leq \sum_{i=1}^k \Pr[(u,v)\in C_i]  & \text{(union bound)} \notag \\
 &\lesssim \sum_{i=1}^k  \frac{d_{H_i}(u,v)}{\Delta} & \text{(part (1) of Lemma \ref{lem:path_q})} \notag \\
 &= \sum_{i=1}^k  \frac{d_{H}(u,v)}{\Delta} & \text{(since $(u,v)\in E(H_1)\cap \ldots \cap E(H_k)$)} \notag \\
 &\lesssim \log n   \frac{d_{H}(u,v)}{\Delta},  & \text{($k\lesssim \log n$)} \notag
\end{align}
and thus $R$ is $O(\log n)$-Lipschitz.

%\medskip

It remains to show that $R$ is $O(\Delta \log n)$-bounded.
Let $x,y\in V(H)$, such that $(x,y)\in R$.
This implies that there exists some path $L$ from $x$ to $y$ in $H\setminus C$.
Let 
\[
X=\bigcup_{P\in {\cal P}} V(P).
\]
If $V(L)\cap X = \emptyset$, then by the basis of the recursion we have that $L$ can contain at most $c=O(1)$ edges, each of length less than $\Delta$, and thus $d_H(x,y) = O(\Delta)$.
Otherwise, $V(L)\cap X \neq \emptyset$.
This implies that there exists some path $P\in {\cal P}$ such that $L$ intersects $P$.
By part (2) of Lemma \ref{lem:path_q} it follows that $d_H(x,y)\lesssim \Delta \log n$.
Thus, in all cases we obtain that $d_H(x,y)\lesssim \Delta \log n$, which implies that $R$ is $O(\Delta \log n)$-bounded, concluding the proof.
To obtain a $\Delta$-bounded, $O(\log^2 n)$-Lipschitz quasipartition, we just scale the quasi-metric space $(V(H), d_H)$ by a factor of $\Theta(1/\log n)$.
\end{proof}

\section{Lipschitz Quasipartitions for Planar Digraphs}
\label{sec:lip}

We now prove our main result on Lipschitz quasipartitions, which is the main technical ingredient of this paper.
We begin by describing the algorithm for computing a quasipartition, and we then prove that it satisfies the desired properties.

The input to the algorithm is a directed planar graph $G$, and some parameter $\Delta>0$.
The algorithm proceeds in the following steps:

\begin{description}
\item{\textbf{Step 1: Sampling from a wave.}}
Pick an arbitrary $v\in V(G)$.
Let ${\cal D}$ be a $(v, \Delta)$-wave, and let $Q'$ be a quasipartition sampled according to ${\cal D}$, induced by some cutset $E'\subseteq E(G)$; that is $Q'\sim {\cal D}$.
Let $\{V_i\}_{i\in \mathbb{N}}$ be as in Definition \ref{defn:wave}.
For any $i$, let $G_i$ be a $(1,\Delta)$-layered graph (with probability 1) obtained by part (iii) of Lemma \ref{lem:wave}.

\item{\textbf{Step 2: Quasipartitioning each layer.}}
By Lemma \ref{lem:layer_lip} it follows that for each $i$, there exists a polynomial-time samplable $\Delta/3$-bounded $O(\log^2 n)$-Lipschitz quasipartition ${\cal D}_i$ of $G_i$.
In polynomial time, we sample some $Q_i\sim {\cal D}_i$, induced by some cutset $F_i$.

\item{\textbf{Step 3: Output.}}
Let 
$F = E' \cup \left(\bigcup_{i} F_i\right)$.
The final output is the quasipartition $R$ induced by the cutset $F$.

\end{description}

We are now ready to prove the main result of this Section.

\begin{proof}[Proof of Theorem \ref{thm:main_lip}]
Let $F$ be the cutset computed by the algorithm described above, and let $R$ be the induced quasipartition.
Let also ${\cal D}$, $E'$, $G_i$, ${\cal D}_i$, $F_i$ be as in the description of the algorithm.

We need to bound the probability that any ordered pair of vertices is not in $R$.
First, consider some $(u,v)\in E(G)$.
By part (i) of Lemma \ref{lem:wave} we have that 
\begin{align}
\Pr[(u,v)\in E'] &\lesssim  \frac{d_G(u,v)}{\Delta} \label{eq:final_cut1}
\end{align}
With probability 1, there exists at most one $i\in \mathbb{N}$, such that $(u,v)\in E(G_i)$.
Conditioned on the event that $(u,v)\in E(G_i)$, by Lemma \ref{lem:layer_lip} we have
\begin{align}
\Pr[(u,v)\in F_i] &\lesssim \log^2 n \frac{d_G(u,v)}{\Delta} \label{eq:final_cut2}
\end{align}
Combining \eqref{eq:final_cut1}, \eqref{eq:final_cut2}, and the fact that $(u,v)$ can be in at most one graph $G_i$, it follows by the union bound that
\begin{align}
\Pr[(u,v)\in F] &= \Pr\left[(u,v)\in E' \cup \left(\bigcup_i F_i\right)\right] \notag \\
&= \Pr\left[\left((u,v)\in E'\right) \vee \left((u,v)\in \bigcup_i F_i\right)\right] \notag \\
&\leq \Pr\left[(u,v)\in E'\right] + \Pr\left[(u,v)\in \bigcup_i F_i\right] & \text{(union bound)} \notag \\
&\lesssim \log^2 n \frac{d_G(u,v)}{\Delta}. \label{eq:final_cut3}
\end{align}

We have thus established that the probability of cutting any particular edge is low.
We next extend this bound to arbitrary ordered pairs of vertices.
Let $u',v'\in V(G)$.
If $v'$ is not reachable from $u'$ in $G$, the condition holds vacuously.
Thus, we may assume w.l.o.g.~that $v'$ is reachable from $u'$ in $G$.
Let $P$ be a shortest path from $u'$ to $v'$ in $G$.
Let $P=w_1,\ldots,w_t$, with $w_1=u'$, $w_t=v'$.
We have that if $(u',v')\notin R$, then at least one edge of $P$ must be in $F$, since otherwise $v'$ is reachable from $u'$ in $G\setminus F$, which implies $(u',v')\in R$. 
Thus, we have
\begin{align}
\Pr[(u',v')\notin R] &\leq \Pr[E(P)\cap F \neq \emptyset] \notag \\
 &= \Pr\left[\left(\bigcup_{i=1}^{t-1} (w_i,w_{i+1})\right) \cap F \neq \emptyset \right] \notag \\
 &= \Pr\left[\bigvee_{i=1}^{t-1}\left((w_i,w_{i+1}) \cap F \neq \emptyset \right) \right] \notag \\
 &\leq \sum_{i=1}^{t-1} \Pr\left[(w_i,w_{i+1}) \cap F \neq \emptyset \right] & \text{(union bound)} \notag \\
 &\lesssim \log^2 n \sum_{i=1}^{t-1} \frac{d_G(w_i,w_{i_1})}{\Delta} & \text{(by \eqref{eq:final_cut3})} \notag \\
 &= \log^2 n \frac{d_G(u',v')}{\Delta} & \text{($P$ is a $u'v'$ shortest path)} \notag
\end{align}
Thus, we have established that the probability that any ordered pair is not contained in the quasipartition is bounded, as required.

It remains to show that $R$ is $\Delta$-bounded.
To that end, it suffices to show that for any $a,b\in V(G)$, if $d_G(a,b)>\Delta$, then $(a,b)\notin R$.
In order to prove the latter property, it suffices to show that for any path $Z$ from $a$ to $b$, at least one edge of $Z$ is in the cutset $F$; that is, $E(Z)\cap F \neq \emptyset$.
Let $Z'$ be prefix of $Z$ of length $\Delta$ (we may assume that $Z'$ has length exactly $\Delta$ by subdividing one edge and inserting a vertex on $Z$ at distance exactly $\Delta$ from $u$).
By part (ii) of Lemma \ref{lem:wave}  we have that, with probability 1, there exists some $i\in \mathbb{N}$, such that $Z'$ can be decomposed into three subpaths $Z'=Z'_1\circ Z'_2\circ Z'_3$, such that 
each $Z'_t$ is contained in some $G_s$.
By setting $Z''$ to be the longest of the three paths $Z'_1$, $Z'_2$ and $Z'_3$, we conclude that there exists some subpath $Z''$ of $Z'$, of length strictly greater than $\Delta/3$, and there exists some $j\in \mathbb{N}$, such that $Z''\subseteq G_j$.
Since, by construction, ${\cal D}_j$ is $\Delta/3$-bounded, it follows that $E(Z'')\cap F_i\neq \emptyset$.
Since $F_i\subseteq F$, and $E(Z'')\subseteq E(Z)$, we get that $E(Z)\cap F\neq \emptyset$, which concludes the proof.
\end{proof}

\bibliographystyle{abbrv}
\bibliography{bibfile}

\appendix

\section{Analysis of Random Exponential Quasiballs}
\label{app:exp_balls}.

\begin{proof}[Proof of Lemma \ref{lem:exp_balls}]
The proof is a direct adaptation of the argument of Bartal~\cite{bartal1996probabilistic}, which holds for the case of metric spaces.
Here, we observe that the same analysis holds for the directed variant we have defined.

Fix some $(u,v)\in E(H)$.
We will derive an upper bound on the probability that $(u,v)$ is included in the cutset $F$.
Let $i\in \{1,\ldots,t\}$.
If $d_H(v_i, u)>d_H(v_i, v)$, then the probability that $(u,v)$ is included at the cutset while computing $X_i$ is 0.
Therefore, we may assume w.l.o.g.~that $d_H(v_i, u)\leq d_H(v_i, v)$.
Following Bartal~\cite{bartal1996probabilistic}, we define the following events:
\begin{description}
\item{$A_i$:} $u,v\in V(H) \setminus (B_{i-1}\cup \ldots \cup B_i)$, where we define $B_0=\emptyset$.
That is, none of $u$ and $v$ have been included in any $X_j$, for any $j<i$.

\item{$M_i^I$:}
$d_H(v_i,v)\leq R_i$, conditioned on $A_i$.
That is, this is the event that both $u$ and $v$ are included in $B_i$, conditioned on $A_i$.

\item{$M_i^X$:}
$d_H(v_i, u)\leq R_i < d_H(v_i, v)$, conditioned on $A_i$.
That is, this is the event that $(u,v)$ is cut when computing $B_i$.

\item{$M_i^N$:}
$R_i < d_H(v_i, u)$, conditioned on $A_i$.
That is, this is the event that none of the vertices $u$ and $v$ is included in $X_i$, conditioned on $A_i$.

\item{$N_i$:}
The event that for all $j\in \{i+1,\ldots,t\}$, we have $|\{u,v\}  \cap  B_j| \leq 1$, conditioned on $A_i$.
That is, the vertices $u$ and $v$ are not both included in the same cluster $B_j$, for all $j>i$, conditioned on $A_i$.
\end{description}

From \cite{bartal1996probabilistic}, we have
\begin{align}
\Pr[M_i^X] &= \int_{d_H(v_i,u)}^{d_H(v_i,v)} p(x) \dd x 
= \frac{n}{n-1} \cdot \left(1-e^{-\frac{d_H(v_i,v)-d_H(v_i,u)}{\Delta}}\right) e^{-d_H(v_i,v)/\Delta} \notag \\
 &\leq \frac{n}{n-1} \cdot \frac{d_H(u,v)}{\Delta} e^{-d_H(v_i,v)/\Delta}. \label{eq:bartal_1}
\end{align}
Moreover
\begin{align}
\Pr[M_i^N] &= \int_{0}^{d_H(v_i,u)} p(x) \dd x = \frac{n}{n-1} \cdot \left(1 - e^{-d_H(v_i, u)/\Delta}\right). \label{eq:bartal_2}
\end{align}

Following \cite{bartal1996probabilistic}, we 
now prove by induction on $i$ that
\[
\Pr[N_i] \leq \left(2-\frac{i}{n-1}\right) \frac{d_H(u,v)}{\Delta}.
\]
We have
\begin{align}
\Pr[N_i] &= \Pr[M_i^X] + \Pr[M_i^N] \Pr[N_{i+1}] \notag \\
 &\leq \frac{n}{n-1} \cdot \frac{d_H(u,v)}{\Delta} e^{-d_H(v_i, u)/\Delta} + \frac{n}{n-1} \left(1- e^{-d_H(v_i,v)/\Delta}  \right) \left(2 - \frac{i+1}{n}\right) \frac{d_H(u,v)}{\Delta} \notag \\
 &\leq \left(2-\frac{t}{n-1}\right) \frac{d_H(u,v)}{\Delta},
\end{align}
which completes the induction.
We have thus established that 
\[
\Pr[(u,v)\in F] \leq \Pr[N_0] \leq 2 d_H(u,v)/\Delta,
\]
which concludes the proof.
\end{proof}

\end{document}